\documentclass[letterpaper, 10 pt, conference]{ieeeconf}  

\IEEEoverridecommandlockouts                              
\usepackage{amssymb}  
\usepackage{mathtools}
\usepackage{algorithmic}
\usepackage{algorithm}
\usepackage{comment}

\newtheorem{definition}{Definition}

\newtheorem{lemma}{Lemma}
\newtheorem{theorem}{Theorem}
\newtheorem{corollary}{Corollary}

\DeclareMathOperator{\diag}{diag}

\title{\LARGE \bf
Coordinating Distributed Energy Resources with Nodal Pricing in Distribution Networks: a Game-Theoretic Approach
}

\author{Eli Brock, Jingqi Li, Javad Lavaei, and Somayeh Sojoudi
\thanks{This material is based upon work supported in part by the U. S. Army Research Laboratory and the U. S. Army Research Office under grant number "W911NF2010219. It was also supported by the Office of Naval research under grant number N000142412673, as well as AFOSR, NSF, and the UC Noyce Initiative.}
\thanks{The authors are with the University of California, Berkeley}%
}

\begin{document}

\maketitle
\thispagestyle{empty}
\pagestyle{empty}

\begin{abstract}
We propose a real-time nodal pricing mechanism for cost minimization and voltage control in distribution networks with autonomous distributed energy resources. Unlike existing methods, the proposed pricing scheme does not require device-aware centralized coordination or communication between prosumers. The resulting market is naturally represented as a stochastic game where prosumers learn feedback control policies to optimize their individual rewards. By developing new sufficient conditions under which a stochastic game is a Markov potential game, we show that the problem of computing an equilibrium for the proposed model is equivalent to solving a single-agent Markov decision process. These new conditions are general and may apply to other applications. An equilibrium is computed for an IEEE test system to empirically demonstrate the near-optimal efficiency of the pricing policy.
\end{abstract}

\section{INTRODUCTION} \label{sec:intro}
\subsection{DER Coordination in Distribution Networks}
An additional 217 GW of distributed energy resources (DERs) is expected on the American electric power grid by 2028, a pace of growth similar to that of bulk generation capacity \cite{razdan_pathways_2025}. The widespread introduction of DERs, which include electric vehicles, heat pumps, storage systems, and distributed solar, marks a critical moment for our energy systems. If operated passively, the extra load from DERs will necessitate expensive infrastructure upgrades and new carbon-intensive fossil-based dispatchable generation. However, with efficient coordination, DER flexibility can improve efficiency on the grid by shaping electric demand to align with intermittent renewable supply and providing local voltage support \cite{alstone_2025_2017}.

There is extensive literature on market mechanisms and control strategies for DER coordination in distribution networks \cite{guerrero_towards_2020}. Coordinated approaches such as distributed optimal power flow, aggregators, and distribution locational marginal pricing (DLMP) assume that a third party orchestrates groups of DERs, either through direct load control or through (shadow)
price incentives, to optimize a single objective function \cite[Section 5]{guerrero_towards_2020}, \cite{chen_wholesale_2024,bai_distribution_2018}. In contrast, peer-to-peer (P2P) markets assume that prosumers control their own devices and optimize their own utility functions; as such, P2P is often analyzed using noncooperative game theory \cite[Section 6]{guerrero_towards_2020}, \cite{chen_energy_2023}. These frameworks make strong assumptions regarding future communication and incentive infrastructures. For example, P2P and distributed optimal power flow methods typically assume that neighboring devices exchange multiple rounds of communication before reaching a collective decision, while DLMP schemes require prosumers\footnote{In active distribution networks, a \textit{prosumer} is a customer who can both consume and produce energy, for example through a rooftop photovoltaic or vehicle-to-grid system.} to share device parameters with a coordinator and schedule future consumption in a rolling-horizon fashion.

Achieving any of the aforementioned frameworks at scale would constitute a major departure from the state of most modern power systems. In practice, distribution system operators (DSOs) usually do not have visibility behind the meter of their customers, and self-interested prosumers autonomously dispatch their own devices without communicating with their neighbors. A more practical mechanism, real-time pricing (RTP), coordinates DERs on the transmission scale by exposing them to time-varying substation-level nodal prices and allowing each prosumer's energy management system (EMS) to optimize their personal electricity usage \cite{mohsenian-rad_optimal_2010}. Most RTP work has focused on optimal control of individual devices as in \cite{wan_model-free_2019} or considers game-theoretic equilibria on the wholesale market as in \cite{ma_decentralized_2013}, without considering the implications for distribution networks.

We propose a new real and reactive nodal pricing structure for minimizing costs and stabilizing voltages on distribution networks. The framework can be understood as a natural extension of RTP to distribution nodes. Unlike traditional DLMPs, the nodal prices are set online, and prosumers do not need to share their device parameters with a DSO. Though the prosumers cannot communicate directly, the nodal pricing scheme allows for indirect coordination through the coupled prices. The resulting networked market is represented as a \textit{stochastic game} (SG) where prosumers attempt to learn closed-loop control policies in the face of uncertain wholesale prices and demand profiles \cite{shapley_stochastic_1953}. To the authors' knowledge, this is the first network-aware distribution-level DER coordination scheme that does not require prosumers to share any information, either between themselves or with a central operator. Next, we derive new, generalizable sufficient conditions under which an SG is a Markov potential game \cite{leonardos_global_2021}, allowing us to compute an equilibrium for the proposed market. Finally, we demonstrate on an IEEE test system that the proposed mechanism results in near-socially-optimal equilibrium policies, despite the potential suboptimality associated with the prosumers' market power.

\subsection{Notation}
$G(p)$ is the geometric distribution with success probability $p$. If a variable or function is defined as $y_i$, then $y$ refers to the vector or vector-valued function collecting all indices. If a set is defined as $\mathcal{X}_i$, then $\mathcal{X}$ refers to the Cartesian product over all indices. If $\mathcal{I}$ is a collection of indices, $y_{\mathcal{I}}$ refers to the elements of the vector $y$ indexed by $\mathcal{I}$. The subscript $-i$ indexes all components except $i$. $\diag(v)$ is the square matrix with the vector $v$ along the diagonal and zeros elsewhere. $\mathbf{1}$ is the vector of all ones, with dimension inferred from context.

\subsection{Stochastic Game Theory Preliminaries} \label{sec:sg}
An infinite-horizon stochastic game (SG) with set of agents $\mathcal{N}$ is a tuple $\mathcal{G}\coloneqq(\mathcal{S}, \{\mathcal{A}_i\}_{i\in\mathcal{N}}, \rho_0, T, \{U_i\}_{i\in\mathcal{N}}, \gamma, \{\Pi_i\}_{i\in\mathcal{N}})$ \cite{shapley_stochastic_1953}. $\mathcal{S}$ is a (possibly infinite) state space, $\mathcal{A}_i$ is the (possibly infinite) action space of agent $i$, $\rho_0:\mathcal{S}\rightarrow\mathbb{R}$ is an initial state distribution, $T:\mathcal{S}\times\mathcal{S}\times\mathcal{A}\rightarrow\mathbb{R}$ is the transition density, $U_i:\mathcal{S}\times\mathcal{A}\rightarrow\mathbb{R}$ is the reward function of agent $i$, $\gamma\in[0,1]$ is the discount factor, and $\Pi_i$ is the set of agent $i$'s available policies. At time $t=0$, the initial state $s_0$ is drawn from the initial state distribution $\rho_0$. At each time $t$, each agent $i$ chooses an action $a_i^t$. Based on the current state and actions, each agent $i$ receives a deterministic reward $U_i(s^t,a^t)$ and the game transitions to the next state according to the transition density function: $s^{t+1} \sim T(\cdot|s^t,a^t)$.

Agent $i$ aims to maximize its infinite-horizon discounted rewards, defined by its \textit{value function}
\begin{equation} \label{eq:value_def}
    V_i^\pi \coloneqq \mathbb{E}_\pi\left[\sum_{t=0}^\infty \gamma^tU_i(s^t,a^t)\right].
\end{equation}
The notation $\mathbb{E}_\pi$ is shorthand for the expectation with respect to $s^0\sim\rho_0$, $a_i^t\sim\pi_i(\cdot|s^t)$ and $s^{t+1} \sim T(\cdot|s^t,a^t)$, where $\pi_i\in\Pi_i$ is agent $i$'s policy. We are now ready to define \textit{Nash equilibria}, the most common solution concept for SGs.

\begin{definition}[Nash Equilibrium \cite{basar_dynamic_1998}]
    A joint policy profile $\pi\in\Pi$ is called a Nash equilibrium of $\mathcal{G}$ if
    \[
        V_i^{\pi_i, \pi_{-i}} \geq V_i^{\tilde{\pi}_i, \pi_{-i}}
    \]
    for all $i\in\mathcal{N}$ and $\tilde{\pi}\in\Pi$.
\end{definition}

In the remainder of the paper, we will use ``equilibrium'' and ``Nash equilibrium'' interchangeably. Intuitively, an equilibrium is a configuration of policies where no agent can unilaterally improve their payoff by changing their policy if all other agents keep their policy fixed. We will often consider \textit{parametric} policy sets of the form $\Pi_i=\{\pi_i^{\theta_i}(\cdot|s^t,a^t): \theta\in\mathbb{R}^{n_\theta}\}$ where $n_\theta$ is the number of parameters. For compact notation, when $\Pi$ is parametric, we will often use $\theta$ in place of $\pi^\theta$. For parametric policy sets, we can also define the notion of a \textit{local} Nash equilibrium.
\begin{definition}[Local Nash Equilibrium]
    If $\Pi$ is a parametric policy class, a set of policy parameters $\theta\in\mathbb{R}^{n_\theta}$ is called a local Nash equilibrium of $\mathcal{G}$ if there exists a $\nu>0$ such that
    \[
        V^{\theta_i, \theta_{-i}} \geq V^{\tilde{\theta}_i, \theta_{-i}}
    \]
    for all $\tilde{\theta}\in\mathbb{R}^{n_\theta}$ such that $\|\tilde{\theta} - \theta\| \leq \nu$.
\end{definition}

Going forward, we will assume that the value functions $V^{\theta_i}$ of parameterized SGs are continuous in $\theta_i$ and differentiable in $\theta_i$ almost everywhere. Local Nash equilibria are stationary points under the \textit{gradient play} algorithm, in which each agent applies the update rule 
\begin{equation} \tag{GP} \label{eq:gp}
    \left(\theta_i^{(k+1)} - \theta_i^{(k)}\right) \propto \nabla_{\theta} V_i^{\theta^{(k)}}.
\end{equation}

An SG with a single agent is known as a \textit{Markov Decision Process} (MDP) and a (local) equilibrium of an MDP is called a (locally) \textit{optimal policy}.

Notice that our definition of SGs explicitly includes the set of feasible policies $\Pi$. Most existing literature on SGs (and MDPs) omits this specification, implicitly assuming that $\Pi$ includes all stochastic policies. For large or infinite state and action spaces, policies are often parameterized in order to design tractable solution methods. However, these parameterizations are treated as function approximations that induce some suboptimality due to their limited expressiveness \cite{agarwal_theory_2021}. While we also use $\Pi$ to encode finite-dimensional parameterizations, in this work, the policy sets are also restricted to capture constraints on the information structure between agents (see Section \ref{sec:policies}). We only seek to find equilibria with respect to the specified policy set without considering the relationship between said equilibria and the equilibria over all stochastic policies.

Computing local equilibria for SGs is hard in general. Unlike in single-objective optimization problems, many natural algorithms such as \ref{eq:gp} may cycle instead of converging to a stationary point \cite{mazumdar_policy-gradient_2020}. In the following, we will use the stochastic game framework to model a distribution grid under a new nodal pricing coordination mechanism (Section \ref{sec:model}). Then, we will develop theory inspired by this application to show that our model belongs to a subclass of SGs for which equilibria can be tractably computed (Section \ref{sec:mpg}).

\section{MODEL} \label{sec:model}
We develop a stochastic game model of a distribution grid with dynamic nodal pricing and autonomous DERs. The proposed market is designed to achieve efficient outcomes for the grid without central coordination or prosumer-to-prosumer communication, as discussed in Section~\ref{sec:intro}.

\subsection{Grid Model} \label{sec:grid_model}
A radial distribution grid is modeled as a directed acyclic graph with a root node $0$ and a set of non-root nodes $\mathcal{N}$. Denote the full set of nodes as $\mathcal{N}^+=\{0\}\cup\mathcal{N}$ and the set of lines and transformers (directed edges) $\mathcal{L}\in\mathcal{N}^+\times\mathcal{N}^+$. By convention, edges are oriented away from the root node. The linear DistFlow model is, for all $i \rightarrow j\in\mathcal{L}$, 
\begin{subequations} \label{eq:distflow}
\begin{align}
    P_{ij} &= p_j + \sum_{k:j \rightarrow k}P_{jk} \\
    Q_{ij} &= q_j + \sum_{k:j \rightarrow k}Q_{jk} \\
    v_i - v_j &= r_{ij}P_{ij} + x_{ij}Q_{ij}
\end{align}
\end{subequations}
where $p_i$, $q_i$ are the real and reactive power consumption of node $i$, $P_{ij}$ and $Q_{ij}$ are the real and reactive power flows on line $i \rightarrow j$, and $v_i$ is the voltage magnitude at node $i$ \cite{baran_network_1989}. The voltage magnitude at the substation $v_0$ is fixed and constant.  Solving the system of equations \eqref{eq:distflow} for $P$, $Q$, and $v$ reveal these quantities to be linear functions of $(p,q)$, where $p$ and $q$ are vectors collecting the nodal injections at the non-root nodes $\mathcal{N}$. We define the matrices $H$, $R$, and $X$ to represent \eqref{eq:distflow} in the compact form
\begin{subequations} \label{eq:distflow_compact}
\begin{align}
    P &= Hp \\
    Q &= Hq \\
    v &= \mathbf{1}v_0 + Rp + Xq.
\end{align}
\end{subequations}
We will continue to use $P$, $Q$, and $v$ to denote these functions, leaving the dependence on $(p,q)$ implicit. At each time, the load profiles $p$ and $q$ are autonomously determined by the prosumers and the DSO incurs the cost
\begin{align} \label{eq:dso_cost}
    C(p, q, \lambda) &= (1-w)\lambda\left(\sum_{i\in\mathcal{N}}p_i + \sum_{i\rightarrow j\in\mathcal{L}}r_{ij}\left(P_{ij}^2+Q_{ij}^2\right)\right) \\
    &+ w\sum_{i\in\mathcal{N}}(v_i-v_0)^2 \nonumber
\end{align}
where $\lambda$ is the wholesale locational marginal price (LMP) at the substation node and $w\in[0,1]$ is a given parameter. The first term in \eqref{eq:dso_cost}, given by the price multiplied by the sum of the loads and the approximate real power losses, captures the cost of importing real power from the wholesale market. While the true real power loss on line $i \rightarrow j$ is $r_{ij}(P_{ij}^2 + Q_{ij}^2)/v_i^2$, we use the fact that $v_i\approx1$ under normal operating conditions to approximate the losses in \eqref{eq:dso_cost}. Since the losses do not appear in the linear model \eqref{eq:distflow}, the DSO cost \eqref{eq:dso_cost} does not account for the small fraction of losses incurred on each line due to losses on downstream lines. We also assume that the distribution network is a price-taker, meaning that it is a small enough participant in the wholesale market that it cannot affect $\lambda$. The second term penalizes deviations from the nominal voltage as in the voltage control literature \cite{zhu_fast_2016}. For simplicity, we assume the nominal voltage across the network is equal to the substation voltage $v_0$. The voltage control weight $w$ controls the trade-off between cost minimization and voltage control, and can be tuned by the DSO until voltages are within an acceptable range.

On distribution networks, the DSO typically handles line ampacity limits through network reconfiguration. Given the complexity introduced by time-varying network topologies, we do not consider line limits; however, this is an important direction for future work.

\subsection{Prosumer Model} \label{sec:prosumer_model}
We model a single prosumer at every non-substation bus, so the set of prosumers is also $\mathcal{N}$. Multiple prosumers per bus may also be handled by the model without affecting the theory. Prosumer $i\in\mathcal{N}$ exhibits inelastic real and reactive power demand $(\bar{p}_i, \bar{q}_i)$ representing the sum of inflexible consumption from devices other than DERs, such as kitchen appliances, lighting, and most other plug loads. Prosumer $i$ also owns a set of flexible DERs $\mathcal{N}_i$ and an EMS enabling automatic intelligent control. These DERs may include heat pumps, electric vehicles (EVs), and energy storage devices. Denoting the consumption from DER $j\in\mathcal{N}_i$ belonging to prosumer $i$ as $\Tilde{p}_{i,j}$, the load from prosumer $i$ is given by the sum of their inelastic and flexible demand:
\begin{equation} \label{eq:load}
(p_i, q_i) \coloneqq (\bar{p}_i, \bar{q}_i) + \sum_{j\in\mathcal{N}_i}(\tilde{p}_{i,j}, \tilde{q}_{i,j}).
\end{equation}
The DERs also have temporal state dynamics:
\begin{equation} \label{eq:der_dynamics}
d^{t+1}_i = f_i(d_i^t, \tilde{p}_i^t, \tilde{q}_i^t, \omega_i^t)
\end{equation}
where $d_i$ is a vector collecting the states of the DERs belonging to prosumer $i$ and $\omega_i^t$ is a random perturbation that will be further discussed in section \ref{sec:exogenous}. Here we introduce the subscript $t$ to index discrete time. The quantities $p$, $q$, $P$, $Q$, $\lambda$, and $v$ also vary in time---the superscript was previously omitted for clarity. The state dynamics $f$ may represent, for example, the state-of-charge of a storage unit/EV or the air/water temperature for a heat pump. The states serve to constrain the DER consumption at each stage:
\begin{equation}\label{eq:der_constraints}
(\tilde{p}_i^t, \tilde{q}_i^t)\in\mathcal{P}_i(d_i^t)
\end{equation}
where we assume $\mathcal{P}_i(\cdot)$ is the feasible set for agent $i$ given $d_i$. Constraint \eqref{eq:der_constraints} may encode state-of-charge, inverter capacity, or comfort constraints.

\subsection{Pricing Mechanism} \label{sec:pricing}
The DSO sets real-time nodal prices for real and reactive power equal to the marginal cost of serving the load at each bus. We assume a \textit{net metering} policy, meaning agents are charged the same rate for net consumption as they are credited for net generation. Specifically, prosumer $i$'s reward function is given by
\begin{align} \label{eq:nodal_price}
    &U_i(d_i,\lambda,\tilde{p},\tilde{q},\bar{p},\bar{q}) \coloneqq u_i(d_i,\tilde{p}_i,\tilde{q}_i) \\
    &- p_i\frac{\partial}{\partial p_i}C(p,q,\lambda) - q_i\frac{\partial}{\partial q_i}C(p,q,\lambda) \nonumber
\end{align}
where $u_i$ is prosumer $i$'s utility function, or their instantaneous benefit from the state-action configuration of their DERs. $u_i$ may encode indoor air temperature preferences or battery degradation costs; for purely shiftable loads, $u_i$ may be set to zero. The last two terms in \eqref{eq:nodal_price} are prosumer $i$'s payment to the DSO. The pricing mechanism in \eqref{eq:nodal_price} is inspired by the DLMP literature \cite{bai_distribution_2018, papavasiliou_analysis_2018}. In these prior works, the DSO computes the prices over a rolling horizon by solving a multi-period scheduling problem for all the devices in the network. By contrast, the prices in \eqref{eq:nodal_price} are set online for the current time period and only require the DSO to meter the current aggregate consumption at each node.

The nodal prices are composed of three components associated with the three terms in \eqref{eq:dso_cost} after expansion: the energy price, the losses price, and the voltage price. The energy price, given by the substation LMP for real power and zero for reactive power, is constant across nodes.  Each prosumer's losses and voltage prices, however, depend on their own consumption as well as that of all agents who share a common ancestor on the network. This coupling introduces strategic interactions and some \textit{market power}, that is, the ability of a participant in a market to manipulate the price. When participants have market power, market equilibria may not maximize social welfare as they do in the fully competitive case. For \eqref{eq:nodal_price} to be an effective coordination mechanism, the gap between the socially optimal outcome and the equilibrium outcome, known as the \textit{price-of-anarchy}, should be small. We empirically verify this in Section \ref{sec:case_study}.

\subsection{Exogenous Quantities} \label{sec:exogenous}
We propose a time-invariant state-space model for the exogenous variables $\lambda$, $\bar{p}$, and $\bar{q}$:
\begin{subequations} \label{eq:exogenous}
\begin{align}
    \label{eq:ex_dynamics}\alpha^{t+1}_i &= g_i(\alpha_i^t,\xi_i^{t+1}) & \forall i\in\{0\}\cup\mathcal{N} \\
    \label{eq:lambda}\lambda^t &= m_0(\alpha_0^t) \\
    \label{eq:demand}(\bar{p}_i^t, \bar{q}_i^t) &= m_i(\alpha_i^t) & \forall i\in\mathcal{N} \\\
    \label{eq:noise}(\xi^t, \omega^t) &\sim \rho_{\xi,\omega}
\end{align}
\end{subequations}
where $g$ is the transition function, $m$ is the measurement function, and $\rho_{\xi, \omega}$ is the noise distribution. The framework $\eqref{eq:exogenous}$ encompasses a rich class of models while satisfying the Markov property. Models of the form \eqref{eq:exogenous} include seasonal autoregressive integrated moving average models, which are common for forecasting time-series econometric data such as electricity demand and prices \cite{durbin_time_2012}. Since the distribution network is assumed to be a price-taker in the wholesale market, we use the simple black-box model in \eqref{eq:exogenous} to capture high-level econometric dynamics instead of including a full optimal-power-flow-based transmission model.

While each prosumer's inelastic demand $\bar{p}$ and the LMP $\lambda$ have separable dynamics in \eqref{eq:exogenous}, they are coupled through the joint distribution $\rho_{\xi,\omega}$, which generates the noise $\xi$. $\xi$ can represent both independent regressors, such as weather, and decoupled perturbations that may simultaneously affect the wholesale price and the demand at different buses. Moreover, since the DER dynamics noise $\omega$ is also generated by $\rho_{\xi, \omega}$, it is correlated with $\xi$ in general. For example, if the DER state is the indoor air temperature, its evolution will be subject to the same perturbations as the ambient outdoor temperature.

\subsection{Policies} \label{sec:policies}
Just as the DSO lacks knowledge of the behind-the-meter devices of its customers, prosumers also cannot see behind the meter of their neighbors. Specifically, we assume that prosumer $i$ observes only their local DER state $d_i$ and has knowledge of their own inelastic demand $(\bar{p}_i, \bar{q}_i)$ and the wholesale LMP $\lambda$, which we assume is public\footnote{The California Independent System Operator, for example, publishes the substation-level real-time-market-clearing LMPs on its website.}. By ``has knowledge'', we mean that the prosumer $i$'s EMS can access all the information it needs to make the best possible prediction of the next quantities $\bar{p}^{t+1}_i$, $\bar{q}^{t+1}_i$ and $\lambda^{t+1}$. In practice, this may include recent histories, periodic information related to the time of day, day-ahead forecasts, and correlated data such as weather---all of which would be accessible to a cloud-connected controller. For our model, given the Markovian structure of \eqref{eq:exogenous}, we achieve the desired information structure by simply allowing agent $i$ to condition its policy directly on $\alpha_i$ and $\alpha_0$, specifically
\begin{equation} \label{eq:local_policies}
a_i^t = \mu_i^{\theta_i}(d_i^t,\alpha_i^t,\alpha_0^t,\eta_i^t)
\end{equation}
where $a_i^t\in\mathbb{R}^2$ is agent $i$'s action, $\eta_i^t\sim\rho_{\eta_i}$ is the policy noise and $\mu_i^{\theta_i}$ is prosumer $i$'s policy, parameterized by $\theta^i$. The generative model \eqref{eq:exogenous}, combined with the local policies \eqref{eq:local_policies}, abstracts away the choice of specific information on which real EMSs might condition their policies and allows for clean theoretical analysis. Optimizing the performance of DERs in a more realistic context without assuming Markovian states and allowing for incomplete observations is an important direction for future research, but falls outside the scope of this work.

To handle the constraint \eqref{eq:der_constraints}, we define the unbounded action spaces $\mathcal{A}_i\coloneqq\mathbb{R}^{2|\mathcal{N}_i|}$, $\forall i\in\mathcal{N}$, and compute the DER consumption by mapping the action onto the feasible set:
\begin{equation} \label{eq:mapping}
(\tilde{p}_i^t, \tilde{q}_i^t) = M_{i}(a_i^t, d_i^t),
\end{equation}
where $M_i:\mathbb{R}^{2|\mathcal{N}_i|}\rightarrow\mathbb{R}^{2|\mathcal{N}_i|}$ is an appropriately chosen function satisfying $M_i(a_i,d_i)\in\mathcal{P}_i(d_i)$ for all possible $a_i$, $d_i$. Depending on the application, $M$ might be a clipping or projection operation. We assume that $M$ is differentiable almost everywhere and continuous.

As an alternative to the local policies \eqref{eq:local_policies}, one could also consider a partially observable stochastic game paradigm, where policies are conditioned on the history of agent $i$'s observations. However, since any given agent may be unaware of their neighbors, it is unlikely that realistic EMSs would attempt to characterize the state of other prosumers' devices. Therefore, we argue that the local policy parameterization \eqref{eq:local_policies} is appropriate for this setting.

\subsection{Stochastic Game Formulation}
We are now ready to formally express the proposed model as an SG, which we call $\mathcal{D}$. The state vector is $s=(x,\alpha)$. The transition density $T$ of $\mathcal{D}$ is characterized by \eqref{eq:der_dynamics}, \eqref{eq:ex_dynamics}, and \eqref{eq:noise}. The reward functions $U_i$, $i\in\mathcal{N}$ are given in \eqref{eq:nodal_price}. The (parametric) joint policy set $\Pi$ is defined according to \eqref{eq:local_policies}. The discount factor $\gamma$ and initial state distribution $\rho_0$ are case-specific.

\section{MARKOV POTENTIAL GAMES} \label{sec:mpg}
In this section, we introduce a class of well-behaved SGs known as Markov Potential Games (MPGs) that admit tractable algorithms for computing equilibria \cite{leonardos_global_2021}. We then present new, generalizable sufficient conditions under which an SG is an MPG and show that the game $\mathcal{D}$ introduced in Section \ref{sec:model} satisfies these conditions.

MPGs generalize the notion of \textit{potential games} in static game theory to SGs \cite{monderer_potential_1996}.
\begin{definition}[Markov Potential Game \cite{zhang_gradient_2024}] \label{def:mpg}
    An SG $\mathcal{G}$ is a Markov potential game if there exists a \textit{potential function} $\phi:\mathcal{S}\times\mathcal{A}\rightarrow\mathbb{R}$ such that
    \[
    V_i^{\pi_i,\pi_{-i}} - V_i^{\tilde{\pi}_i,\pi_{-i}} = \Phi^{\pi_i,\pi_{-i}} - \Phi^{\tilde{\pi}_i,\pi_{-i}}
    \]
    where 
    \[
    \Phi^\pi \coloneqq \mathbb{E}_\pi\left[\sum_{t=0}^\infty\gamma^t\phi(s^t,a^t)\right]
    \]
    for all $i\in\mathcal{N}$ and $\pi,\tilde{\pi}\in\Pi$.
\end{definition}

In other words, an MPG is an SG where each agent's value function is characterized by a single \textit{potential value function} $\Phi^\pi$, given by the discounted sum of the potential function $\phi$.
An immediate consequence of Definition \ref{def:mpg} is that joint policy profiles that (locally) optimize the potential function $\Phi^\pi$ are also (local) Nash equilibria.

If a game is an MPG and the potential function $\phi$ is known, then finding a (local) equilibrium is reduced to (locally) solving a single-agent MDP for the optimal joint policy with rewards given by $\phi$. Unlike multi-agent SGs, MDPs can be reliably solved by methods from reinforcement learning and dynamic programming. In particular, for parameterized policy classes, a local equilibrium of an MPG can be found by the gradient ascent algorithm
\begin{equation} \tag{$\Phi$-GA} \label{eq:ga}
    \left(\theta^{(k+1)} - \theta^{(k)}\right) \propto \nabla_{\theta} \Phi^{\theta^{(k)}}.
\end{equation}
Definition \ref{def:mpg} implies that, for MPGs, \ref{eq:ga} is equivalent to the decentralized gradient play algorithm \ref{eq:gp}; see \cite[Proposition B.1]{leonardos_global_2021} for details.

Definition \ref{def:mpg} is challenging to verify. It may be natural to suspect that an SG is an MPG if it is potential at each stage, that is, if there exists $\phi$ such that
 \begin{equation} \label{eq:stagewise}
 U_i(s,a_i,a_{-i}) - U_i(s,\tilde{a}_i,a_{-i}) = \phi(s,a_i,a_{-i}) - \phi(s,\tilde{a}_i,a_{-i})
 \end{equation}
 for all $i\in\mathcal{N}$, $a,\tilde{a}\in\mathcal{A}$, and $s\in\mathcal{S}$. Unfortunately, this is not true; see \cite[Proposition 2]{zhang_gradient_2024} for a counterexample. SGs satisfying \eqref{eq:stagewise} are, however, the most promising candidates for MPGs since static game theory techniques may be used to find a candidate potential function $\phi$ given the rewards $U_i$. As such, the existing literature has focused on obtaining sufficient conditions under which SGs satisfying \eqref{eq:stagewise} are MPGs \cite{leonardos_global_2021, zhang_gradient_2024, mguni_learning_2021}.

\cite{mguni_learning_2021} claims to establish a sufficient condition for MPGs under an assumption, called ``state transitivity'', that the reward functions are also potential with respect to the state. This assumption is quite restrictive and is not satisfied by the game of interest $\mathcal{D}$. \cite[Prop 3.2, C1]{leonardos_global_2021} and \cite[Lemma 8]{zhang_gradient_2024} both introduce what we call the \textit{action-independent transitions} (AIT) sufficient condition requiring that the transition density is independent of the agents' actions. However, this restriction precludes any controllable state dynamics and is only satisfied by repeated one-shot games. \cite[Prop 3.2, C2]{leonardos_global_2021} generalizes AIT, but it is difficult to check and is specific to SGs with finite state and action spaces.

\cite[Lemma 8]{zhang_gradient_2024} introduces what we call the \textit{local states and policies} (LSP) sufficient condition. LSP holds when 1) each state is owned by a certain agent, 2) the rewards are potential in states as well as actions, 3) each local state space has its own conditionally independent transition density, and 4) each agent's policy is conditioned only on its local state. LSP does not apply to $\mathcal{D}$ because the state space includes the component $\alpha_0$ that does not belong to any agent, and the agents' local state transitions are coupled through the joint distribution $\rho_{\xi,\omega}$.

While neither AIT nor LSP applies directly to $\mathcal{D}$, $\mathcal{D}$ exhibits features of each. Similar to AIT, the exogenous states $\alpha$ evolve independently of the agents' actions. As in LSP, the agents condition their policies on private local states with transition densities that are independent of other agents' actions. We now present a new, verifiable sufficient condition generalizing both AIT and LS that applies to the proposed distribution system game $\mathcal{D}$.

\begin{theorem} \label{thm:main}
    An SG $\mathcal{G}$ is an MPG with potential function $\phi$ if, for each agent $i\in\mathcal{N}$, there exists a \textit{local state space} $\mathcal{S}_i$ such that $\mathcal{S}=\mathcal{S}_i\times\mathcal{S}_{-i}$ and
    \begin{enumerate}
        \item\label{it:stagewise} $\phi$ is a stagewise potential function for both the states and actions, that is
        \begin{align} \label{eq:stagewise_potential}
        &U_i(s_i,s_{-i},a_i,a_{-i}) - U_i(\tilde{s}_i,s_{-i},\tilde{a}_i,a_{-i}) \\
        &= \phi(s_i,s_{-i},a_i,a_{-i}) - \phi(\tilde{s}_i,s_{-i},\tilde{a}_i,a_{-i}) \nonumber
        \end{align}
        for all $s_i,\tilde{s}_i\in\mathcal{S}_i$, $s_{-i}\in\mathcal{S}_{-i}$, and $a,\tilde{a}\in\mathcal{A}$.
        \item\label{it:independent_policy} Other agents' policies do not depend on the local states, that is $\pi_{-i}(a_{-i}|s_i,s_{-i})=\pi_{-i}(a_{-i}|s_{-i})$.
        \item\label{it:marginal_transition} The marginal transition density of the non-local states, defined as
        \[
        T_{-i}(s_{-i}'|s,a)\coloneqq\int_{s_i'\in\mathcal{S}_i}T(s_i',s_{-i}'|s,a)ds_i'
        \]
        does not depend on agent $i$'s action or local state, that is $T_{-i}(s_{-i}'|s_i,s_{-i},a_i,a_{-i}) = T_{-i}(s_{-i}'|s_{-i},a_{-i})$.
    \end{enumerate}
\end{theorem}

\begin{proof}
    See appendix.
\end{proof}

Theorem 1 generalizes LSP by allowing the agents' non-local state spaces to overlap and by relaxing the requirement that the local state transition probabilities be conditionally independent. While Theorem \ref{thm:main} is motivated by the power systems setting considered here, it is quite general and may extend to other applications. The idea of the proof is to show that the difference between the potential value function and agent $i$'s reward is independent of its policy. At a high level, this is achieved by combining and generalizing the fundamental features of AIT and LSP.

Referring back to Section \ref{sec:model}, we apply Theorem \ref{thm:main} to show that $\mathcal{D}$ is an MPG. From \eqref{eq:distflow_compact}, define the matrix
\begin{align*}
   L(\lambda) &= \lambda(1-w)\begin{bmatrix}
        H^T\diag(r)H &  \mathbf{0} \\
        \mathbf{0} & H^T\diag(r)H
    \end{bmatrix} \\
    &+ w\begin{bmatrix}
        R^T \\ X^T
    \end{bmatrix}\begin{bmatrix}
        R & X
    \end{bmatrix}
\end{align*}

\begin{corollary} \label{cor:application}
$\mathcal{D}$ is an MPG with potential function 
    \[
    \phi_{\mathcal{D}}(d,a,\lambda)=\sum_{i\in\mathcal{N}}u_i(d_i,\tilde{p}_i,\tilde{q}_i) - \tilde{C}(p,q,\lambda)
    \]
where 
\begin{equation} \label{eq:potential_social_cost}
    \tilde{C}(p,q,\lambda) = C(p,q,\lambda) + \sum_{i\in\mathcal{N}}\begin{bmatrix}
        p_i \\ q_i
    \end{bmatrix}^TL(\lambda)_{\mathcal{I}_i\mathcal{I}_i}\begin{bmatrix}
        p_i \\ q_i
    \end{bmatrix}
\end{equation}
and $\mathcal{I}_i$ = $\{i,i+|\mathcal{N}|\}$ for all $i\in\mathcal{N}$.
\end{corollary}

\begin{proof}
    See appendix.
\end{proof}

Corollary \ref{cor:application} is instructive. Ideally, the agents would cooperate to maximize the discounted sum of the social benefit $\sum_{i\in\mathcal{N}}u_i(d_i,\tilde{p}_i) - C(p, q,\lambda)$. Instead, due to the market power associated with the losses and voltage prices, they minimize the discounted sum of $\phi_\mathcal{D}$, which includes the the additional second term in \eqref{eq:potential_social_cost}. In Section \ref{sec:case_study}, we will demonstrate that the effect of this additional term is small in practice.

\section{EXPERIMENT} \label{sec:case_study}
We evaluate the efficiency of the proposed market on a benchmark IEEE network populated with autonomous storage units. The equilibrium policies are compared with a more naive non-nodal RTP pricing structure as well as the theoretical socially optimal policies. Our analysis is from a mechanism design perspective, meaning that we are only interested in evaluating equilibria and do not seek to model the learning process. As such, we use centralized computation to compute an equilibrium and assume that, in practice, prosumers will find it in the process of optimizing their local policies.

\subsection{Algorithm}
For MPGs, the potential value function is simply the discounted sum of the potential function. Therefore, the problem of finding a local equilibrium reduces to the problem of computing a locally optimal policy for an MDP. Given that we have access to the transition density and reward functions, we choose to apply a simple
\textit{value gradient} algorithm inspired by \cite{heess_learning_2015} instead of more popular model-free policy gradient algorithms. Value gradient algorithms depend on the reparameterization trick, where the policy and transition densities are expressed as deterministic functions of independent random variables:
\begin{subequations} \label{eq:reparameterization}
\begin{align}
a^t &= \mu_\theta(s^t,a^t,\eta^t) \\
s^{t+1} &= h(s^t,a^t,\zeta^{t+1}) \\
\label{eq:sampling}\eta^t&\sim\rho_\eta, ~ \zeta^t\sim\rho_\zeta
\end{align}
\end{subequations}
For the game of interest $\mathcal{D}$, the transition density defined in \eqref{eq:der_dynamics}, \eqref{eq:ex_dynamics}, and \eqref{eq:noise} as well as the policies \eqref{eq:local_policies} are already written in reparameterized form; we reuse $\mu$ and $\eta$ given the one-to-one correspondence between their use in the general case and their use in $\mathcal{D}$. Given an MPG with potential function $\phi$, define the estimated potential value function
\begin{equation} \label{eq:estimator}
\hat{\Phi}^\theta(s^0,H,\zeta,\eta) \coloneqq \sum_{t=0}^H\phi(s^t,a^t)
\end{equation}
where $s^t$, $t\geq1$ and $a^t$ are generated by \eqref{eq:reparameterization}. $\hat{\Phi}^\theta$ is an unbiased estimate of the potential value function
\[
\mathbb{E}\left[\hat{\Phi}^\theta(s^0,H,\zeta,\eta)\right] = \Phi^\theta
\]
where the expectation is taken with respect to \eqref{eq:sampling}, $s^0\sim\rho_0$, and $H \sim G(1-\gamma)$. By extension, the gradient of the estimator is an unbiased estimate of the gradient of the potential value function:
\[
\mathbb{E}\left[\nabla_\theta\hat{\Phi}^\theta(s^0,H,\zeta,\eta)\right] = \nabla_\theta\Phi^\theta.
\]
This leads to Algorithm \ref{alg:sga} for computing local equilibria, which is equivalent to stochastic gradient ascent (SGA) on the potential value function.
\begin{algorithm}[t!]
\caption{SGA} \label{alg:sga}
\begin{algorithmic}[1]
    \REQUIRE S-MPG $\mathcal{G}$, $\gamma$, $\rho_\xi$, $\rho_\eta$, $\rho_0$, $\gamma$, $\beta$, $N_\text{train}$, $N_\text{batch}$
    \STATE Define $\hat{\Phi}^\theta$ from $\mathcal{G}$
    \STATE Arbitrarily initialize $\theta$
    \FOR{$1$ \TO $N_{\text{train}}$}
    \FOR{$1$ \TO $N_{\text{batch}}$}
    \STATE Sample $s^0\sim\rho_0$, $H \sim G(1-\gamma)$, $\zeta_t\sim\rho_\zeta$, $\eta_t\sim\rho_\eta$
    \STATE\label{step:gradient} $\hat{\nabla}\gets\hat{\nabla} + \frac{1}{N_\text{batch}}\nabla_\theta \hat{\Phi}^\theta(s_0,H,\zeta,\eta)$
    \ENDFOR
    \STATE $\theta\gets\theta + \beta\hat{\nabla}$
    \ENDFOR
    \RETURN $\theta$
\end{algorithmic}
\end{algorithm}
Step \ref{step:gradient} can be computed by backpropogation through the Bellman equation as in \cite[Sec. 4.1]{heess_learning_2015}; automatic differentiation software such as PyTorch can perform this computation off-the-shelf.

The update rule in Step \ref{step:gradient} is equivalent to the stochastic gradient play update rule. As $N_{\text{batch}}\rightarrow\infty$, Algorithm \ref{alg:sga} recovers the deterministic gradient ascent algorithm \eqref{eq:ga}, which is in turn equivalent to the gradient play algorithm \eqref{eq:gp} as discussed in Section \ref{sec:mpg}. As established, the scope of this work is limited to computing and evaluating a local equilibrium in the interest of mechanism design. However, consider a more practical multi-agent reinforcement learning setting where agents compute unbiased estimates of their policy gradients in a decentralized, model-free manner using the policy gradient theorem \cite{sutton_policy_1999}. In expectation (with large batch sizes), this multi-agent policy gradient algorithm is equivalent to Algorithm \ref{alg:sga}. Therefore, Algorithm \ref{alg:sga} may mirror the learning process of multi-agent systems when the sufficient conditions in Theorem \ref{thm:main} hold. See \cite{zhang_global_2020} for a convergence analysis of policy gradient algorithms for continuous state and action spaces (single-agent policy gradient results suffice in the MPG setting).

\subsection{Setup}
We demonstrate the proposed distribution grid market on the IEEE 18-bus test case from \cite{grady_application_1992} with 1-hour timesteps. Each prosumer owns a single battery energy storage unit, such as a Tesla Powerwall, with its internal state being the current state-of-charge. A storage unit's charging is constrained by its maximum state-of-charge as well as its inverter capacity. Customer utility functions are set to zero, so they focus only on price arbitrage. The substation LMP and inelastic demand profiles are simple noisy sinusoids with 1-day periods, with the demand profiles staggered such that they peak (on average) six hours before the LMP. Loads are scaled until losses reach approximately 10\% of the total power flow, consistent with a heavily loaded distribution network. The local policies \eqref{eq:local_policies} are Gaussian with the mean parameterized as a simple affine function of the states. The experimental setup is further detailed in
the appendix.

\subsection{Performance}
We train three joint policies for comparison. All three policies are trained by Algorithm \ref{alg:sga} with $N_{\text{batch}}=1$, $N_{\text{train}}=500$, $\beta=.001$, and the policy parameters $\theta$ initialized to the zero vector. The trained policies are:
\begin{enumerate}
    \item Equilibrium (EQ) policies, trained to minimize the potential value function by setting $\phi=\phi_\mathcal{D}$ in \eqref{eq:estimator}.
    \item Socially optimal (SO) policies, computed by setting $\phi$ in \eqref{eq:estimator} to the social welfare value function
    \[
    W^\theta \coloneqq \mathbb{E}_\pi\left[\sum_{t=0}^\infty\gamma^t\left(\sum_{i\in\mathcal{N}}u_i(d_i^t,\tilde{p}_i^t)-C(p^t,\lambda^t)\right)\right].
    \]
    \item Policies under uniform pricing (UN), trained identically to the EQ policies but with the network removed by setting $r = x = 0$.
\end{enumerate}
\begin{figure}[t!]
    \centering
    \includegraphics[width=\linewidth]{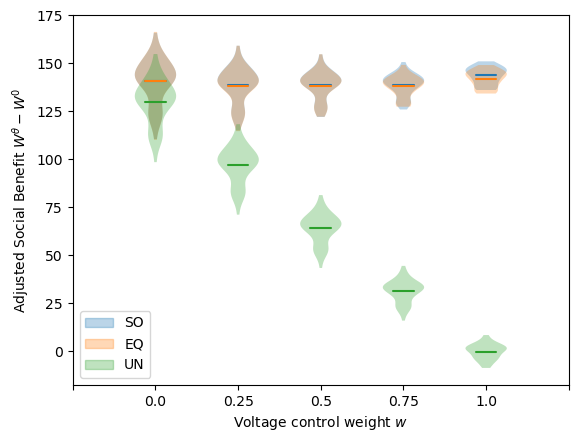}\vspace{-1em}
    \caption{Relative performance of the EQ, SO, and UN policies. EQ is almost equivalent to SO and is more socially beneficial than UN pricing regardless of the tradeoff between import costs and voltage control.
    }\vspace{-1em}
    \label{fig:performance}
\end{figure}
\begin{figure}[t!]
    \centering
    \includegraphics[width=\linewidth]{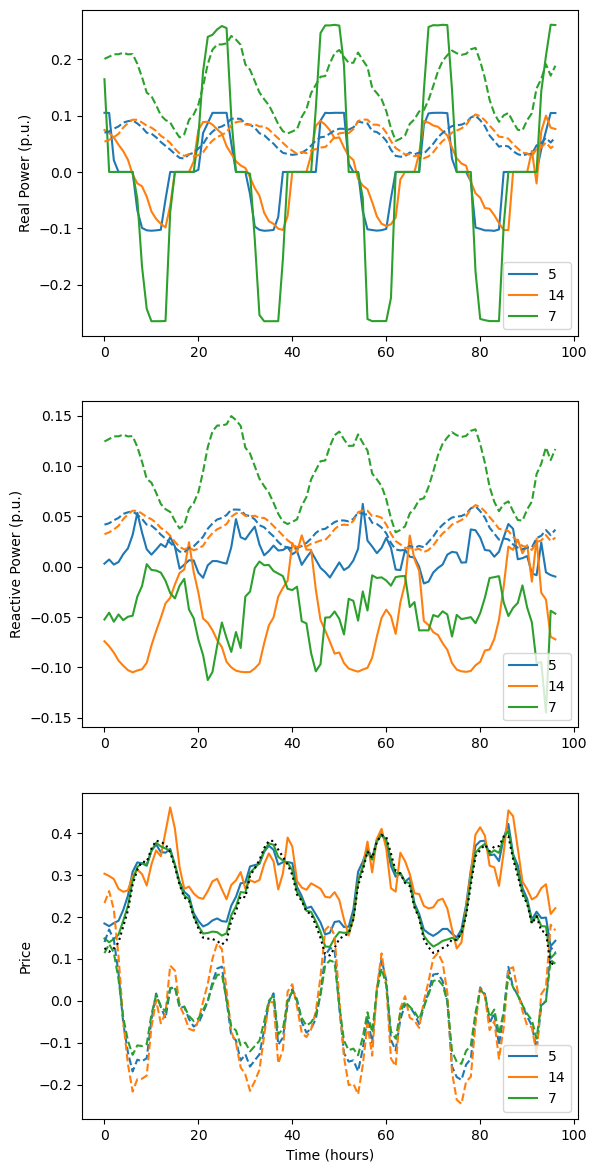}\vspace{-1em}
    \caption{Equilibrium behavior of three random agents over four-day rollout.}\vspace{-1em}
    \label{fig:demo}
\end{figure}
Fig. \ref{fig:performance} shows the approximate adjusted social welfare value function $W^\theta - W^0$ for each policy across five values of the voltage control weight $w$\footnote{Because lower values of $w$ deprioritize voltage control, the losses approximation in \eqref{eq:dso_cost} may only be well-justified for sufficiently high $w$. The low $w$ cases in Fig. \ref{fig:performance} are meant to explore the strategic implications of the game, rather than to serve as realistic simulations.}. $W^0$ is the social welfare value function when the prosumers idle their storage units at every timestep. The adjustment is necessary because $W^0$ appears as a constant term in the expanded form of $W^\theta$, so including it would distort the relative performance of the policies. The approximate social welfare value function is the empirical average over 50 discounted rollouts of length 500 ($.99^{500} < .01)$. To allow for a one-to-one comparison, the same 50 trajectories of the random variables $\xi$ and $\omega$ were used for all 15 evaluations. Fig. \ref{fig:performance} visualizes the distributions of the 50 rollouts and their means, indicated by horizontal lines. All policies converged by the end of the 500 training rollouts.

UN is the equilibrium outcome when the DSO ignores the network and simply forwards the LMP to each prosumer. In this aggregated context, the prices no longer vary by node, so there is no market power. The EQ policies learn to respond to nodal prices reflecting network-related social costs. However, in theory, they may exploit their market power to benefit themselves at the expense of their neighbors, thereby failing to meaningfully increase social welfare compared to the network-blind UN policies. The results in Fig. \ref{fig:performance} indicate that this is not the case: the EQ policies outperform the UN policies by at least 8\% across all five voltage control weights.

The SO policies capture the ficticious best-case-scenario where all devices cooperatively maximize social welfare instead of responding to prices. Fig. \ref{fig:performance} shows that the price-of-anarchy, given by the gap between EQ and SO, is negligible relative to the gap between EQ and UN
\footnote{While it is difficult to see from the figure, SO did indeed slightly outperform EQ for all five voltage control weights, as expected.}. For this benchmark network, the takeaway is that the benefits of the proposed nodal pricing mechanism outweigh the costs---in fact, it is almost as efficient as the best possible coordination mechanism given the local policy structure.

\subsection{Demonstration}
To build intuition, we include a four-day simulation of the deterministic EQ policies with $w=0.75$ in Fig. \ref{fig:demo}. Three randomly chosen agents are presented. The first and second subplots show the real and reactive power load, with inelastic demand as a dashed line and storage consumption as a solid line. The third subplot shows the nodal prices for real and reactive power as solid and dashed lines, respectively, with the substation LMP included as a dotted black line. The nodal pricing mechanism creates different incentive structures and, by extension, different behaviors between the prosumers. Fundamentally, the agents must balance arbitraging over the substation LMP, minimizing losses, and stabilizing voltage magnitudes, given the apparent power capacity of their inverters.

During the four-day simulation, voltage magnitudes fluctuated between $\pm0.10$ p.u. and network losses ranged from about 2\% and 12\% of the total power flow through the substation.

\section{CONCLUSION}
We propose a practical pricing scheme for DER coordination in distribution networks, accounting for both import costs and voltage stability. In order to compute equilibria for the resulting model, we characterize new generalizable sufficient conditions under which a stochastic game is a Markov potential game and prove that our application satisfies these new conditions. Finally, the proposed mechanism is shown to be efficient on a benchmark distribution network. Interesting directions for future research include accounting for the affect of the aggregate distribution network load on the LMP and rigorously bounding the price-of-anarchy.



\bibliographystyle{ieeetr}
\bibliography{references}

\newpage
\begin{appendix}

\subsection{Additional Notation}
$I_k$ is the identity matrix of dimension $k$. $\mathbf{1}$, $\mathbf{0}$, and $e_k$ are the matrix of all ones, the matrix of all zeros, and the $k$th standard basis vector, respectively, with dimensions inferred from context. $U(\underline{x},\overline{x})$ is the uniform distribution between its $\underline{x}$ and $\overline{x}$. $N(\mu,\Sigma)$ is the multivariate normal distribution with mean $\mu$ and covariance $\Sigma$.

\subsection{Omitted Proofs} \label{sec:proofs}

\subsubsection{Proof of Theorem \ref{thm:main}}

For each agent $i\in\mathcal{N}$, define the \textit{dummy function}
\begin{equation} \label{eq:dummy_def}
\psi_i(s,a) = U_i(s,a) - \phi(s,a).
\end{equation}
Sufficient condition \ref{it:stagewise} implies that the dummy function does not depend on the local state or action:
\begin{equation} \label{eq:independent_dummy}
\psi_i(s_i,s_{-i},a_i,a_{-i}) = \psi_i(s_{-i},a_{-i}).
\end{equation}
Combining \eqref{eq:value_def}, \eqref{eq:dummy_def}, and \eqref{eq:independent_dummy}, agent $i$'s value function can be decomposed as follows
\begin{equation} \label{eq:value_decomposition}
    V_i^\pi = \mathbb{E}_\pi\left[\sum_{t=0}^\infty\gamma^t\phi(s^t,a^t)\right] + \mathbb{E}_\pi\left[\sum_{t=0}^\infty\gamma^t\psi_i(s_{-i}^t,a_{-i}^t)\right].
\end{equation}
Notice that the first term is the desired potential value function $\Phi^\pi$ from Definition \ref{def:mpg}. To satisfy Definition \ref{def:mpg}, we need to show that the second term in \eqref{eq:value_decomposition} does not depend on $\pi_i$. First, bring the expectation inside the summation:
\[
\mathbb{E}_\pi\left[\sum_{t=0}^\infty\gamma^t\psi(s_{-i}^t,a_{-i}^t)\right] = \sum_{t=0}^\infty\gamma^t\mathbb{E}_\pi\left[\psi(s_{-i}^t,a_{-i}^t)\right].
\]
Clearly, it suffices to show that $\mathbb{E}_\pi\left[\psi(s_{-i}^t,a_{-i}^t)\right]$ is independent of $\pi_i$ for all $t$. For compactness, we write $\int_{s_i\in\mathcal{S}_i}\cdot ds_i$ simply as $\int_{s_i}\cdot ds_i$, and likewise for the integral over the nonlocal state space $\mathcal{S}_{-i}$ and relevant subsets of the action space.
\begin{align*}
    &\mathbb{E}_\pi\left[\psi(s_{-i}^t,a_{-i}^t)\right] \\
    =& \int_{s}\Pr_\pi(s^t=s)\int_{a_{-i}}\pi_{-i}(a_{-i}|s)\psi_i(s_{-i},a_{-i})da_{-i}ds \\
    =& \int_{s_{-i}}\Pr_\pi(s_{-i}^t=s_{-i})\int_{s_i}\Pr_\pi(s_i^t=s_i|s_{-i}^t=s_{-i}) \\
    &\int_{a_{-i}}\pi_{-i}(a_{-i}|s)\psi_i(s_{-i},a_{-i})da_{-i}ds_ids_{-i}
\end{align*}
where $\Pr_\pi(\cdot)$ is shorthand for the probability with respect to where $s^0\sim\rho_0$, $a^t\sim\pi(\cdot|s^t)$, and $s^{t+1} \sim T(\cdot|s^t,a^t)$. We now apply sufficient condition \ref{it:independent_policy} to rewrite $\pi_{-i}$ without the dependence on $s_i$:
\begin{align*}
    =& \int_{s_{-i}}\Pr_\pi(s_{-i}^t=s_{-i})\int_{s_i}\Pr_\pi(s_i^t=s_i|s_{-i}^t=s_{-i}) \\
    &\int_{a_{-i}}\pi_{-i}(a_{-i}|s_{-i})\psi_i(s_{-i},a_{-i})da_{-i}ds_ids_{-i} \\
    =& \int_{s_{-i}^t}\Pr_\pi(s_{-i}^t=s_{-i})\int_{a_{-i}^t}\pi_{-i}(a_{-i}^t|s_{-i}^t)\psi_i(s_{-i}^t,a_{-i}^t) \\
    &\int_{s_i^t}\Pr_\pi(s_i^t|s_{-i}^t)ds_i^tda_{-i}^tds_{-i}^t \\
    =& \int_{s_{-i}^t}\Pr_\pi(s_{-i}^t=s_{-i}) \\
    &\int_{a_{-i}^t}\pi_{-i}(a_{-i}^t|s_{-i}^t)\psi_i(s_{-i}^t,a_{-i}^t)da_{-i}^tds_{-i}^t
\end{align*}
The only remaining nominal dependence on $\pi$ is through the term $\Pr_\pi(s_{-i}^t=s_{-i})$. We show that this term also does not depend on $\pi$ by induction. Suppose that $\Pr_\pi(s_{-i}^t=s_{-i})$ is independent of $\pi$ for some $t$. In general, we have
\begin{align*}
    \Pr_\pi(s_{-i}^{t+1} = s_{-i}') =& \int_{s}\Pr_{\pi}(s^t=s)\int_{a}\pi(a^t|s^t) \\
    &\int_{s_i'}T(s_i',s_{-i}'|s,a)ds_i'dads
\end{align*}
Applying sufficient condition \ref{it:marginal_transition} gives
\begin{align*}
    =& \int_{s}\Pr_{\pi}(s^t=s)\int_{a}\pi(a|s)T_{-i}(s_{-i}'|s_{-i},a_{-i})da_ida_{-i}ds
\end{align*}
Decomposing the inner integral,
\begin{align*}
    =& \int_{s}\Pr_{\pi}(s^t=s)\int_{a_{-i}}\pi_{-i}(a_{-i}|s)T_{-i}(s_{-i}'|s_{-i},a_{-i}) \\
    &\int_{a_i}\pi_i(a_i|s^t)da_ida_{-i}ds \\
    =& \int_{s}\Pr_{\pi}(s^t=s)\int_{a_{-i}}\pi_{-i}(a_{-i}|s)T_{-i}(s_{-i}'|s_{-i},a_{-i})da_{-i}ds
\end{align*}
Decomposing the outer integral,
\begin{align*}
    =& \int_{s_{-i}}\Pr_{\pi}(s_{-i}^t=s_{-i})\int_{s_i}\Pr_\pi(s_i^t=s_i|s_{-i}^t=s_{-i}) \\
    &\int_{a_{-i}}\pi_{-i}(a_{-i}|s)T_{-i}(s_{-i}'|s_{-i},a_{-i})da_{-i}ds_ids_{-i}
\end{align*}
Applying sufficient condition \ref{it:independent_policy},
\begin{align*}
    =& \int_{s_{-i}}\Pr_{\pi}(s_{-i}^t=s_{-i})\int_{s_i}\Pr_\pi(s_i^t=s_i|s_{-i}^t=s_{-i}) \\
    &\int_{a_{-i}}\pi_{-i}(a_{-i}|s_{-i})T_{-i}(s_{-i}'|s_{-i},a_{-i})da_{-i}ds_ids_{-i} \\
    =& \int_{s_{-i}}\Pr_{\pi}(s_{-i}^t=s_{-i})\int_{a_{-i}}\pi_{-i}(a_{-i}|s_{-i})T_{-i}(s_{-i}'|s_{-i},a_{-i}) \\
    &\int_{s_i}\Pr_\pi(s_i^t=s_i|s_{-i}^t=s_{-i})ds_ida_{-i}ds_{-i} \\
    =& \int_{s_{-i}}\Pr_{\pi}(s_{-i}^t=s_{-i}) \\
    &\int_{a_{-i}}\pi_{-i}(a_{-i}|s_{-i})T_{-i}(s_{-i}'|s_{-i},a_{-i})da_{-i}ds_{-i} \\
\end{align*}
The only remaining nominal dependence on $\pi_i$ is through the term $\Pr_\pi(s_{-i}^t=s_{-i})$, which is independent of $\pi_i$ by the inductive hypothesis. For the base case $t=0$, note that $\Pr_\pi(s_{-i}^0=s_{-i})$ does not depend on $\pi$ since the initial state is drawn directly from the distribution $\rho_0$.

\subsubsection{Proof of Corollary \ref{cor:application}}
We will use the following lemma.
\begin{lemma} \label{lemma:quad_potential}
    For any vector $v\in\mathbb{R}^\ell$, matrix $Q\in\mathbb{R}^{\ell \times \ell}$, and partition of indices $\mathcal{I}=\{\mathcal{I}_1,\dots,\mathcal{I}_J\}$ such that $\bigcup_{k=1}^J\mathcal{I}_k=[\ell]$ and $\mathcal{I}_i\cap\mathcal{I}_j=\emptyset$ for all $i\neq j$,
    \[
    \nabla_{v_{\mathcal{I}_i}}\left[v_{\mathcal{I}_i}^T\nabla_{v_{\mathcal{I}_i}}(v^T Q v)\right] = \nabla_{v_{\mathcal{I}_i}}\left(v^TQv+\sum_{j=1}^Nv_{\mathcal{I}_j}^TQ_{\mathcal{I}_j\mathcal{I}_j}v_{\mathcal{I}_j}\right).
    \]
\end{lemma}

\begin{proof}
    The expressions are shown to be equivalent through expansion. To avoid clutter, we will write the subscript $\mathcal{I}_i$ as simply $i$. We will use the $E_i$ to denote the matrix composed of the rows of the identiy matrix indexed by $\mathcal{I}_i$. Beginning with the left-hand side:
    \begin{align*}
        &\nabla_{v_i}\left[v_i^T\nabla_{v_i}(v^T Q v)\right] \\
        &= E_i\nabla\left[(E_iv)^TE_i\nabla(v^TQv)\right] \\
        &= E_i\nabla\left[v^TE_i^TE_i(Q+Q^T)v\right] \\
        &= E_i\left[E_i^TE_i(Q+Q^T) + (Q+Q^T)E_i^TE_i\right]v \\
        &= E_i(Q+Q^T)(I+E_i^TE_i)v
    \end{align*}
    where we use the fact that $E_iE_i^T=I$. For the right-hand side:
    \begin{align*}
        &\nabla_{v_i}\left(v^T Qv+\sum_{j\in\mathcal{I}}v_j^TQ_{jj}v_j\right) \\
        &= E_i\nabla\left[v^T \left(Q+\sum_{j\in\mathcal{I}}E_j^TE_jQE_j^TE_j\right)v\right] \\
        &=E_i\left[(Q+Q^T) + \sum_{j\in\mathcal{I}}E_j^TE_j(Q+Q^T)E_j^TE_j\right]v \\
         &= E_i(Q+Q^T)(I+E_i^TE_i)v
    \end{align*}
    where we use the fact that $E_iE_j=0$ when $i \neq j$.
\end{proof}

We now prove Corollary \ref{cor:application} by checking the sufficient conditions from Theorem \ref{thm:main} in sequence, starting with condition \ref{it:stagewise}. For each agent $i\in\mathcal{N}$, define the local state $s_i=d_i$. The property \eqref{eq:stagewise_potential} holds if and only if $\phi_\mathcal{D}(d,\lambda,a) - U_i(d,\lambda,a)$ does not depend on the local state $d_i$ or the local action $a_i$. Beginning with the local state, we have
\begin{align*}
    &\phi_\mathcal{D}(d,\lambda,a) - U_i(d,\lambda,a) =  \sum_{j\in\mathcal{N}}u_j(d_j,\tilde{p}_j,\tilde{q}_j) - \tilde{C}(p,q,\lambda) \\
    &- \left(u_i(d_i,\tilde{p}_i,\tilde{q}_i) - p_i\frac{\partial}{\partial p_i}C(p,q,\lambda) - q_i\frac{\partial}{\partial q_i}C(p,q,\lambda)\right) \\
    &= \sum_{j \neq i}u_j(d_j,\tilde{p}_j,\tilde{q}_j) - \tilde{C}(p,q,\lambda) \\
    &- \left(- p_i\frac{\partial}{\partial p_i}C(p,q,\lambda) - q_i\frac{\partial}{\partial q_i}C(p,q,\lambda)\right)
\end{align*}
which does not include $d_i$ as desired. We now want to show the final expression does not depend on the action $a_i$ by way of $(\tilde{p}_i,\tilde{q}_i)$ or, by extension, $(p_i,q_i)$. Clearly, the summation of the other agents' utilities does not depend on $a_i$, so we need only consider the other terms where $(p_i,q_i)$ appears. Checking the gradient of these terms:
\begin{align*}
    &\nabla_{(p_i,q_i)}\left(p_i\frac{\partial}{\partial p_i}C(p,q,\lambda) + q_i\frac{\partial}{\partial q_i}C(p,q,\lambda) - \tilde{C}(p,q,\lambda)\right) \\
    &= \nabla_{(p_i,q_i)}\left[\begin{bmatrix}
        p_i \\ q_i
    \end{bmatrix}^T\nabla_{(p_i,q_i)}\left(\begin{bmatrix}
        p \\ q
    \end{bmatrix}^TL(\lambda)\begin{bmatrix}
        p \\ q
    \end{bmatrix}\right)\right] \\
    &- \nabla_{(p_i,q_i)}\left[\begin{bmatrix}
        p \\ q
    \end{bmatrix}^TL(\lambda)\begin{bmatrix}
        p \\ q
    \end{bmatrix} + \sum_{j\in\mathcal{N}}\begin{bmatrix}
        p_j \\ q_j
    \end{bmatrix}^TL_{\mathcal{I}_j\mathcal{I}_j}(\lambda)\begin{bmatrix}
        p_j \\ q_j
    \end{bmatrix}\right] \\
    &= 0.
\end{align*}
The first step uses the fact that
\[
C(p,q,\lambda) = (w-1)\lambda\sum_{i\in\mathcal{N}}p_i + \begin{bmatrix}
    p \\ q
\end{bmatrix}^T L(\lambda) \begin{bmatrix}
    p \\ q
\end{bmatrix},
\]
which can be checked from \eqref{eq:distflow_compact} and \eqref{eq:dso_cost}. The linear term cancels trivially and the second step applies Lemma \eqref{lemma:quad_potential}. We can conclude from the resulting equality that $\phi_\mathcal{D}(d,\lambda,a) - U_i(d,\lambda,a)$ does not depend on $a_i$, satisfying condition \ref{it:stagewise}.

The local policy of agent $j$ is given by
\[
\pi_j^{\theta_j}(a_j|s) = \rho_{\eta_j}\left(\left\{\eta_j:a_j=\mu_j^{\theta_j}(d_j,\alpha_j,\alpha_0,\eta_j)\right\}\right)
\]
When $j \neq i$, $d_i$ does not appear in this expression, so condition \ref{it:independent_policy} is satisfied.

By construction of the local state space, we have $s_{-i}=(d_{-i},\alpha)$. The marginal transition density of the nonlocal states is given by
\begin{align*}
    T_{-i}(s_{-i}'|s,a) = \rho_{\xi,\omega}(\{&(\xi,\omega): \\
    &d_j'=f_j(d_j,\tilde{p}_j,\tilde{q}_j,\omega_j) \quad \forall j\neq i, \\
   &\alpha_i' = g_j(\alpha_j,\xi_j) \quad \forall j\in\mathcal{N}\cup\{0\}\})
\end{align*}
Since neither the local state $d_i$ nor the local action $a_i$ (by way of $(\tilde{p}_i,\tilde{q}_i)$ or $(p_i,q_i)$) appear in this expression, condition \ref{it:marginal_transition} is satisfied and the proof is complete.

\subsection{Experimental Setup} \label{sec:setup}
A scalar state variable $d_i$ represents storage unit $i$'s state-of-charge. The state dynamics \eqref{eq:der_dynamics} are given by
\[
f_i(d_i^t,\tilde{p}_i^t) = d_i^t + \tilde{p}_i^t
\]
where $\tilde{p}_i^t$ is also a scalar, since the prosumer owns only a single DER. Note that the storage dynamics include no stochastic component (represented in the general game $\mathcal{G}$ by $\omega$). The feasible set \eqref{eq:der_constraints} is given by
\[
\mathcal{P}_i(d_i) = \{(\tilde{p}_i,\tilde{q}_i):-d_i\leq\tilde{p}_i\leq\overline{d_i}-d_i,\tilde{p_i}^2 + \tilde{q}_i^2 \leq b_i^2\}
\]
where $b_i$ is the apparent power inverter capacity and $\overline{d_i}$ is the maximum state-of-charge. We choose $M_i$ in \eqref{eq:mapping} to be a lazy projection operation that first clips the real component fo $a_i$ into the range $[-d_i, \overline{d_i}]$ and then projects the resulting action onto the 2-norm ball of radius $b_i$. It is also possible to perform a true projection, but doing so requires solving a convex optimization problem at every timestep and therefore takes much longer to train. Since the authors observed nearly identical results using the true and lazy projections, we choose the latter for ease of reproducibility.

The agents' utility functions $u_i$ are set to zero, so they will only seek to minimize their utility bills given their fixed inelastic demand. We do not include round-trip inefficiencies or battery degradation since the purpose of the example is to demonstrate the fundamental strategic features of the proposed model, but such details can easily be incorporated under the general framework in Section \ref{sec:prosumer_model}.

The agent policies are Gaussian with fixed variance where the mean is an affine function of the states:
\begin{subequations}
\begin{align} \label{eq:affine_policy}
    \mu_i^{\theta_i}(d_i^t,\alpha_i^t,\alpha_0^t,\eta_i^t) &= \theta_i^{d_i}d_i^t + \theta_i^{\alpha_i}\alpha_i^t + \theta_i^{\alpha_0}\alpha_0^t + \theta_i^0 + \eta_i^t \\
    \eta_i^t &\sim N(\mathbf{0},\Sigma_{\eta_i}).
\end{align}
\end{subequations}
While more complex neural-network-based policies could also be employed, we found that affine policies are sufficiently expressive for this simple application and exhibited more reliable training.

The inelastic demand and LMP profiles are sinusoids with periods of one day perturbed by normally-distributed random noise. Specifically, the dynamics \eqref{eq:ex_dynamics} are given by
\begin{align*}
    g_i(\alpha_i^t,\xi_i^{t+1}) &= A\alpha_i^t + B\xi_i^t \\
    \lambda^t &= \lambda^*(1 + \kappa m^T\alpha_0^t) \\
    (\bar{p}_i^t, \bar{q}_i^t) &= (\bar{p}_i^*, \bar{q}_i^*)(1 + \kappa m^T\alpha_0^t) \\
    \xi_i^t &\sim N(0,\Sigma_\xi)
\end{align*}
where
\begin{align*}
    A &= \begin{bmatrix}
        \begin{matrix}
            \cos\frac{\pi}{12} & -\sin\frac{\pi}{12} \\
        \sin\frac{\pi}{12} & \cos\frac{\pi}{12}
        \end{matrix} & \mathbf{0} \\
        \mathbf{0} & \begin{matrix}
            \mathbf{0} & 0 \\
            I_{\tau} & \mathbf{0}
        \end{matrix}
    \end{bmatrix} \\
    B &= \begin{bmatrix}
        0 & 0 & e_1^T
    \end{bmatrix}^T \\
    m &= \begin{bmatrix}
        1 & 0 & \sigma_\xi\mathbf{1}^T
    \end{bmatrix} \\
    \Sigma_\xi &= z\mathbf{1}\mathbf{1}^T + \diag((1-z)\mathbf{1}).
\end{align*}
$\kappa$ is the amplitude factor of the sinusoids. $\lambda^*$ and $\tilde{p}_i^*$ are the mean LMP and agent inelastic demand, respectively. $\tau$ is the noise duration, $\sigma_\xi$ is the standard deviation factor for the exogenous parameters, and $z$ is the correlation coefficient.

The initial state distribution is defined by
\begin{align*}
    d_i^0 &\sim U(0,\overline{d_i}) \\
    \alpha_i^0 &= \begin{bmatrix}
        \cos\frac{2\pi}{24}(t_0 + \delta_i) & \sin\frac{2\pi}{24}(t_0 + \delta_i) & 0 & 0 & 0
    \end{bmatrix} \\
    t_0 &\sim U(0,23) \\
    \delta_i &\sim U(\underline{\delta}, \overline{\delta}) \quad \forall i\in\mathcal{N} \\
    \delta_0 &= 0
\end{align*}
where $t_0$ represents the starting hour of the day of the simulation and $\delta_i$ is the phase difference between agent $i$'s inelastic demand and the LMP, randomized between agents to introduce heterogeneity.

The network model is the 18-bus radial distribution system from \cite{grady_application_1992}, one of seven distribution benchmark cases distributed with MATPOWER. One storage-equipped agent is located at each of the 15 load buses. To approximate a heavily-loaded distribution network with losses on the order of $10\%$, the nominal loads are tripled.

$\bar{p}^*$ and $\bar{q}^*$ are set to the (tripled) real and reactive nominal load values from the MATPOWER case file and $\lambda^*$ is set to 1. Branch parameters $r$ and $x$ are also taken from the MATPOWER case file. We further set $\kappa_0=1/2$, $\sigma_\xi=1/10$, $z=.9$, $\tau=3$, $\underline{\delta}=3$, and $\overline{\delta}=9$. For the storage units, we set $b_i=1.5\sqrt{\bar{p}_i^{*2} + \bar{q}_i^{*2}}$, $\overline{d_i} = 6\overline{\tilde{p}_i}$, and $\Sigma_{\eta_i}=\diag((\bar{p}_i^{*2},\bar{q}_i^{*2}))$. Finally, we set the discount factor $\gamma=0.99$.

\end{appendix}


\end{document}